\documentclass[reqno,12pt]{amsart}

\usepackage[margin=1in]{geometry}
\usepackage{amsmath,amssymb,amsthm,mathtools}
\usepackage[dvipsnames]{xcolor}
\usepackage{enumitem}
\usepackage{graphicx}
\usepackage{tikz}
\colorlet{LinkBlue}{MidnightBlue!99!black}
\usepackage[colorlinks=true,linkcolor=LinkBlue,citecolor=LinkBlue,urlcolor=LinkBlue,hypertexnames=false]{hyperref}
\usepackage[all]{hypcap}
\usepackage{aliascnt}
\usepackage[nameinlink,noabbrev]{cleveref}
\usetikzlibrary{arrows.meta,calc,patterns}
\definecolor{figaccent}{HTML}{1BAF7A} 
\newcommand{\samplegraph}{%
    \coordinate (x) at (0,0);      \coordinate (p1) at (1,0);
    \coordinate (p2) at (2,0.9);   \coordinate (p3) at (2,-0.9);
    \coordinate (p4) at (3,0);     \coordinate (p5) at (4,0);
    \coordinate (y) at (4.9,0.75); \coordinate (r) at (4.9,-0.75);}
\newcommand{\samplegraphedges}{(x)--(p1) (p1)--(p2) (p1)--(p3) (p2)--(p4) (p3)--(p4) (p4)--(p5) (p5)--(y) (p5)--(r)}

\usepackage{autonum}

\newcommand{\R}{\mathbb{R}}

\newcommand{\supp}{\operatorname{supp}}
\newcommand{\C}{\mathbb{C}}
\newcommand{\epsi}{\varepsilon}

\newcommand{\norm}[1]{\lVert #1 \rVert}
\newcommand{\HH}{\mathcal{H}}

\newcommand{\DD}{\mathcal{D}}

\newcommand{\defeq}{\stackrel{\mathrm{def}}{=}}
\newcommand{\1}{\mathbf{1}}

\newtheorem{theorem}{Theorem}[section]
\newaliascnt{proposition}{theorem}
\newtheorem{proposition}[proposition]{Proposition}
\aliascntresetthe{proposition}
\newaliascnt{lemma}{theorem}
\newtheorem{lemma}[lemma]{Lemma}
\aliascntresetthe{lemma}
\newaliascnt{corollary}{theorem}
\newtheorem{corollary}[corollary]{Corollary}
\aliascntresetthe{corollary}
\theoremstyle{definition}
\newaliascnt{definition}{theorem}
\newtheorem{definition}[definition]{Definition}
\aliascntresetthe{definition}
\theoremstyle{definition}
\newaliascnt{example}{theorem}

\aliascntresetthe{example}
\theoremstyle{remark}
\newaliascnt{remark}{theorem}

\aliascntresetthe{remark}
\numberwithin{equation}{section}

\crefname{theorem}{theorem}{theorems}
\Crefname{theorem}{Theorem}{Theorems}
\crefname{proposition}{proposition}{propositions}
\Crefname{proposition}{Proposition}{Propositions}
\crefname{lemma}{lemma}{lemmas}
\Crefname{lemma}{Lemma}{Lemmas}
\crefname{corollary}{corollary}{corollaries}
\Crefname{corollary}{Corollary}{Corollaries}
\crefname{definition}{definition}{definitions}
\Crefname{definition}{Definition}{Definitions}
\crefname{example}{example}{examples}
\Crefname{example}{Example}{Examples}
\crefname{remark}{remark}{remarks}
\Crefname{remark}{Remark}{Remarks}
\crefname{equation}{equation}{equations}
\Crefname{equation}{Equation}{Equations}

\makeatletter
\def\l@subsection{\@tocline{2}{0pt}{3pc}{5pc}{}}
\makeatother

\begin{document}

\title{Lieb--Robinson bounds via the Combes--Thomas estimate}

\author{Tali Braester}
\author{Haochen Luo}
\author{George Peykanu}
\author{Sergii Solomakha}
\address{(T.~Braester, H.~Luo, G.~Peykanu, S.~Solomakha) Department of Mathematics, University of Washington, Seattle, WA 98195, USA}
\email{tcbrae@uw.edu, haochl4@uw.edu, gpeyka@uw.edu, ssoloma@uw.edu}
\begin{abstract}

This paper studies quantum spin systems: networks of spins coupled by local interactions. The Lieb--Robinson bound predicts that in these systems, information propagates at an essentially finite speed. We produce a new proof of this result. Our approach relies on a spectral bound on entries of resolvents of band matrices, the Combes--Thomas estimate, rather than on the standard path-counting argument. 
\end{abstract}
\maketitle

\section{Introduction}
Quantum spin systems are discrete, fundamental models of many-body quantum mechanics. While direct interactions involve only neighboring sites, cumulating them as in the Schr\"odinger propagator connects arbitrarily far sites. This raises a fundamental question:
how quickly does information propagate through quantum spin systems? 
\par This problem was first addressed for translation-invariant systems by Lieb and Robinson in 1972 \cite{LiebRobinson}, who proved that information in quantum spin systems  propagates at a speed that is bounded independently of the system size, up to an exponentially small remainder. Robinson extended the result to interactions whose strength decays exponentially and need not be translation invariant \cite{Robinson76}. The subject was significantly revived when Hastings \cite{Hastings04} used the bound to prove a multi-dimensional Lieb--Schultz--Mattis theorem. This demonstrated that Lieb--Robinson bounds can be exploited to describe physical properties of quantum materials beyond the speed of propagation.  Nachtergaele \& Sims \cite{NachtergaeleSims06} extended the bound to exponentially decaying interactions on discrete metric spaces, and Hastings \& Koma \cite{HastingsKoma} further extended it to power-law decaying interactions and obtained constants independent of the local Hilbert-space dimension. 
\par Such results, collectively known in the literature as Lieb--Robinson bounds, show that information propagates at an (essentially) finite speed. For finite-range interactions
the speed is bounded by a constant depending only on the interaction strength and on the number of neighbors.
Specifically, the influence of a spin on distant spins is exponentially small outside a ball of radius that grows linearly with time. In this paper, we provide a new spectral approach to Lieb--Robinson bounds, based on the Combes--Thomas estimate \cite{CombesThomas} instead of the standard path-counting argument of \cite{Hastings04, HastingsKoma, NachtergaeleSims06}.

\subsection{Setup and basic notation}\label{sec:setup}
We model a finite quantum spin system by a connected graph $G=(\Lambda, E)$ with at least one edge, and a Hamiltonian $H : \mathcal H \rightarrow \mathcal H$ acting on the state space of the full system, the Hilbert space
\[
\mathcal H
= 
\bigotimes_{x\in\Lambda}\mathbb C^2.
\]
We provide $G$ with the graph distance $d$: given $x,y \in \Lambda$, $d(x,y)$ is the minimum length of a path on the graph joining $x$ and $y$. The \emph{degree} $\deg(x)$ of a site $x\in\Lambda$ is the number of edges containing $x$. 

In this description, each $x\in \Lambda$ carries the two-dimensional Hilbert space $\C^2$, corresponding to a quantum bit (qubit); qubits connected by an edge directly communicate (see \Cref{fig:spin-system}). The Hamiltonian then takes the form
\begin{align}
H=\sum_{e\in E} H_e, \qquad  \text{where:}
\end{align}
\begin{itemize}
    \item $H_e : \HH \rightarrow \HH$ is a self-adjoint operator;
    \item $H_e$ is supported on the endpoints $x,y$ of $e$ (see \Cref{def:support} below for the support of an operator);
    \item the operator norm of $H_e$ satisfies $\|H_e\|\le 1$.
\end{itemize}

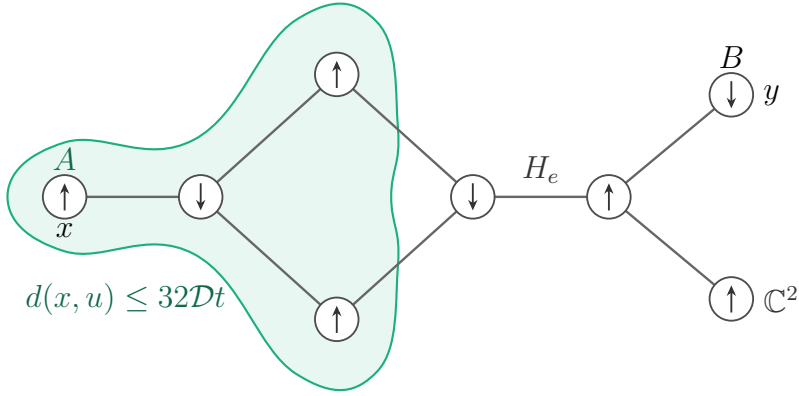
\begin{figure}[ht]
\centering
\begin{tikzpicture}[x=1.8cm,y=1.8cm,>={Stealth[length=4pt,width=3.5pt]}]
    \samplegraph
    \filldraw[fill=figaccent!10,draw=figaccent,line width=0.8pt]
        plot[smooth cycle,tension=0.75] coordinates {(-0.42,0) (0,0.4) (0.95,0.42) (1.7,1.3) (2.3,1.28) (2.45,0.5) (2.4,0) (2.45,-0.5) (2.3,-1.28) (1.7,-1.3) (0.95,-0.42) (0,-0.4)};
    \node[figaccent!60!black,anchor=north] at (0.45,-0.55) {$d(x,u)\le 32\DD t$};
    \draw[black!60,line width=0.9pt] \samplegraphedges;
    \node[black!80,above=1pt] at ($(p4)!.5!(p5)$) {$H_e$};
    \foreach \v/\s in {x/1,p1/-1,p2/1,p3/1,p4/-1,p5/1,y/-1,r/1} {
        \filldraw[fill=white,draw=black!70,line width=0.7pt] (\v) circle (0.16);
        \draw[->,black!80,line width=0.7pt] ($(\v)+(0,-0.1*\s)$) -- ($(\v)+(0,0.1*\s)$);
    }
    \node[black!80,right=8pt] at (r) {$\C^2$};
    \node[below=6pt] at (x) {$x$};
    \node[above=6pt,figaccent!60!black] at (x) {$A$};
    \node[right=8pt] at (y) {$y$};
    \node[above=6pt] at (y) {$B$};
\end{tikzpicture}%
\caption{A quantum spin system: a qubit $\C^2$ at each site, an interaction $H_e$ on each edge $e$, an operator $A$ supported at $x$ and $B$ at $y$. Shaded: the light cone $\{u : d(x,u)\le 32\DD t\}$ of \Cref{thm:main}, here of radius $2$, in which $A(t)$ is approximately supported.}
\label{fig:spin-system}
\end{figure}

In the Heisenberg picture, an operator $A$ evolves as $A(t) \defeq e^{iHt}A\,e^{-iHt}$. We aim to describe the extent to which an operator $A$ initially localized at $x$ affects site $y$ at time $t$ as a function of $t$ and of $d(x,y)$. We introduce the following tool:

\begin{definition}[Support of an operator]\leavevmode\phantomsection\label{def:support}%
The \emph{support} of an operator $A\colon(\C^2)^{\otimes\Lambda}\to(\C^2)^{\otimes\Lambda}$, denoted $\supp(A)$, is the smallest set $S\subseteq\Lambda$ for which
\begin{equation}\label{eq:support-def}
A = a\otimes\1_{\Lambda\setminus S},
\qquad a\colon(\C^2)^{\otimes S}\to(\C^2)^{\otimes S}.
\end{equation}
\end{definition}

 The support is not a robust notion, because it is unstable under small perturbation: if $A$ has support $S$ and $X$ has support $\Lambda$, then $A+\epsi X$ generally has support $\Lambda$ even for small $\epsi$. In practice, one measures how close an operator is to one with known support via commutators: an operator supported in $S$ commutes with every operator supported in $\Lambda \setminus S$; conversely, an operator that commutes with every operator supported in $\Lambda \setminus S$ is supported in $S$. \Cref{thm:approximate-support} below is a quantitative version of this converse. There, the commutators are merely small rather than vanishing, and the support is correspondingly approximate rather than exact.  See \cite[eqs.~(1.1)--(1.2) and Lemma~2.1]{Nachtergaele13} for a proof.

\begin{proposition}[Approximate support]\label{thm:approximate-support}
Let $\varepsilon>0$, $S \subseteq \Lambda$, and $A:(\C^2)^{\otimes\Lambda}\to(\C^2)^{\otimes\Lambda}$. If $\norm{[A,B]} \le \varepsilon$ for every $B:(\C^2)^{\otimes\Lambda}\to(\C^2)^{\otimes\Lambda}$ with $\|B\|\le 1$ and $\supp(B)\subseteq \Lambda \setminus S$, then $A$ is ``approximately'' supported on $S$; that is, there exists an operator $A'$ with $\supp(A')\subseteq S$ such that $\norm{A-A'} \le \varepsilon$.
\end{proposition}

These observations allow us to relate the speed of information propagation to the norm of the commutator
\begin{equation}\label{commutator-def}
    \bigl\|[A(t),B]\bigr\|,
\end{equation}
where $A$ is an operator initially supported at a site $x\in \Lambda$ and $B$ is an arbitrary operator supported at some $y\in \Lambda$.

\subsection{The Lieb--Robinson bound}\label{sec:lr-bound} Suppose we have a quantum spin system, represented by a graph $G=(\Lambda, E)$ and a Hamiltonian $H$. For a fixed time $t\ge 0$, is the norm \eqref{commutator-def} small when $d(x,y)$ is large? 
Such a bound would give control on the speed of information propagation in the quantum spin system, up to a small remainder.

Lieb and Robinson \cite{LiebRobinson} were the first to give a positive answer to this question. Later works extended their result to broader classes of quantum spin systems and sharpened the form of the bound \cite{Robinson76, Hastings04, HastingsKoma, NachtergaeleSims06}; see \cite{Chen} for a survey. Following \cite[Section~5]{NachtergaeleSims}, any such decay estimate, uniform in the size of the graph, is informally called a Lieb--Robinson bound.

For the quantum spin systems described in \Cref{sec:setup}, such bounds can be brought to the form \cite[Section~3]{Chen}: if $A$ is an operator with initial support localized to a site $x\in \Lambda$, $A(t)$ is its Heisenberg-evolved operator, and $B$ is any operator supported on a site $y\in \Lambda$, then, for every $t \ge 0$,
\begin{equation}\label{eq:lr-bounds}
    \bigl\lVert [A(t), B] \bigr\rVert
      \;\le\; C_1\, \mathrm{e}^{\,C_2t - C_3d(x,y)} \|A\| \| B \|,
\end{equation}
where $C_1, C_2, C_3 \in \mathbb{R}^+$ are constants depending only on the maximal degree of the graph, and not on $|\Lambda|$ or on the operators $A$, $B$.

\par For instance, Hastings' bound \cite[Eq.~(4)]{Hastings04} and the bound of \cite[Proposition~3.5]{Chen} have a factorial form, which can be converted to the exponential form \eqref{eq:lr-bounds} by a standard estimate \cite[Example~3.6]{Chen}. The Nachtergaele--Sims bound appears in this form as Eq.~(5.6) of their lecture notes \cite{NachtergaeleSims}. Our work rederives \eqref{eq:lr-bounds} in the following form:
\begin{theorem}[Lieb--Robinson bound]\leavevmode\phantomsection\label{thm:main}%
    Let $A$ be an operator with support at a site $x\in \Lambda$ and $A(t)$ its Heisenberg-evolved operator. Then for any operator $B$ supported at a site $y\in \Lambda$ and every $t \ge 0$,
    \begin{equation}\label{eq:main-result}
    \bigl\lVert [A(t), B] \bigr\rVert
      \;\le\; 18\, \DD\, \mathrm{e}^{\,32\DD t - d(x,y)} \|A\| \|B\|,
\end{equation}
where $\DD = \max_{x \in \Lambda} \deg(x)$ is the maximal degree of the graph.
\end{theorem}
\par For an interpretation of \eqref{eq:main-result}, note that the norm of the commutator $[A(t),B]$ decays at least exponentially when $d(x,y) > 32\DD t$. Informally, in
the language of \Cref{thm:approximate-support}, $A(t)$ is approximately
supported in the \emph{light cone} $\{u \in \Lambda : d(x,u) \le 32\DD t\}$:
information originating at $x$ propagates no faster than the velocity
$32\DD$. Equivalently, up to an exponentially small error, an operation
performed at $x$ at time $t = t_0$ cannot be detected at any time
$t > t_0$ outside the light cone
$\{u \in \Lambda : d(x,u) \le 32\DD(t - t_0)\}$. For an illustration,
see \Cref{sec:example} and \Cref{fig:spin-system}.

\par The bound \eqref{eq:main-result} is an instance of \eqref{eq:lr-bounds} with $C_1 = 18\DD$, $C_2 = 32\DD$ and $C_3 = 1$. The constants in \eqref{eq:main-result} depend on the graph only through its maximal degree $\DD$, but not through its size $|\Lambda|$. The velocity $32\DD$ is linear in the maximal degree $\DD$, as in \cite[Proposition~3.5]{Chen}, with a larger constant; as noted in \cite[Section~3.2.3]{Chen}, the linear scaling itself cannot be improved.
\par As an intermediate result, the proof produces a superexponential bound outside the light cone, similar to Hastings' factorial bound \cite[Eq.~(4)]{Hastings04}:
\begin{equation}\label{eq:factorial}
    \bigl\|[A(t),B]\bigr\|
    \;\le\; 6\DD \left(\frac{8e\DD t}{d(x,y)-1}\right)^{d(x,y)-1} \|A\|\,\|B\|
    \qquad \text{for } d(x,y) > 32\DD t + 1,
\end{equation}
from which \eqref{eq:main-result} follows by the same standard estimate \cite[Example~3.6]{Chen}. Since $(ct)^n/n! \le (ect/n)^n$, the factorial bounds mentioned above imply bounds of the form \eqref{eq:factorial}; in the case of \cite[Proposition~3.5]{Chen}, the constant is $4(\DD-1)$ in place of our $8\DD$.

\par We emphasize that the bound \eqref{eq:main-result} itself is not new, and we have made no attempt to optimize its constants. Our main contribution is a different proof: the standard path-counting argument (see, e.g., \cite{Hastings04, HastingsKoma, NachtergaeleSims06}) is replaced by a resolvent estimate, the Combes--Thomas estimate \cite{CombesThomas}, together with a contour integral. To our knowledge, this is the first time the Combes--Thomas estimate has been used to prove a Lieb--Robinson bound for a quantum spin system. On the other hand, Combes--Thomas estimates and contour integrals are  part of the classical toolbox in the analysis of single-particle systems, see e.g. \cite[Chapter~10]{AizenmanWarzel} and the references given there.
Therefore the argument presented here offers single-particle mathematical physicists a bridge to the many-body world. It is also arguably less technical than the path-counting arguments.
We outline the proof of \Cref{thm:main} in \Cref{sec:proof-strategy}.

\subsection{An example}\label{sec:example}

To illustrate \Cref{thm:main}, we simulate the Heisenberg spin chain \cite{Heisenberg28, Bethe31}. We consider an open chain of $11$ qubits at the sites
$\Lambda = \{0, 1, \dots, 10\}$, with Hamiltonian
\[
  H = \sum_{i=0}^{9} H_{(i,i+1)},
  \qquad
  H_{(i,i+1)} = \tfrac{1}{3}\bigl(X_iX_{i+1} + Y_iY_{i+1} + Z_iZ_{i+1}\bigr).
\]
Here $X_j, Y_j, Z_j$ act at site $j \in \{0, 1, \dots, 10\}$ as the
Pauli matrices
\[
  X = \begin{pmatrix} 0 & 1 \\ 1 & 0 \end{pmatrix}, \qquad
  Y = \begin{pmatrix} 0 & -\mathrm{i} \\ \mathrm{i} & 0 \end{pmatrix}, \qquad
  Z = \begin{pmatrix} 1 & 0 \\ 0 & -1 \end{pmatrix}.
\]
Note that each edge term has operator norm $\|H_{(i,i+1)}\| = 1$, and
the maximal degree of the chain is $\DD = 2$. Take $A = X_0$ (supported at $x = 0$), and $B = X_y$ (supported at a site $y \in \{1, \dots, 10\}$). 
\Cref{fig:light-cone} shows the norm $\|[A(t),B]\|$ as a function of
the time $t$ and the distance $d(x,y)$.

\begin{figure}[ht]
\centering
\includegraphics[width=\linewidth]{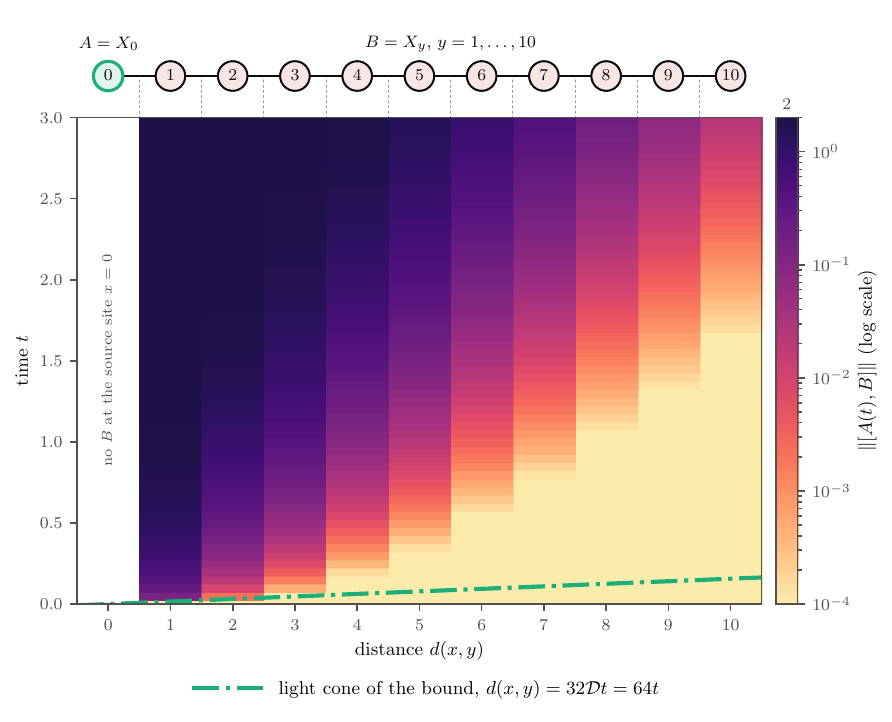}
\caption{Numerical simulation of $\|[A(t),B]\|$, via exact
diagonalization, for the Heisenberg chain of \Cref{sec:example}.
In the plot, 
values of $\|[A(t),B]\|$ below
$10^{-4}$ appear in the lightest shade.
The green line is the light cone $d(x,y) = 32\DD t = 64t$: according to 
\eqref{eq:main-result}, the
right-hand side of \eqref{eq:main-result} decays exponentially below it (i.e., outside the cone). This shows that our result is not optimal: the commutator norm can also be small above the line.
}
\label{fig:light-cone}
\end{figure}

\subsection{Proof strategy and connection to existing results}\label{sec:proof-strategy}

Our proof initially follows the approach of \cite{LiebRobinson}, which we review in \Cref{sec:2}; then in \Cref{sec:3}, we opt for spectral arguments that are, to the best of our knowledge, novel in the many-body setting. 

\par For each edge $e\in E$ we introduce
\[
F(e,t)
\defeq
\sup_{\substack{\supp(X)\subseteq e, \\ \|X\|\le 1}}
\bigl\|[X,A(t)]\bigr\|.
\]
As proved in \cite{LiebRobinson} and refined in \cite{Hastings04, HastingsKoma}, $F$ satisfies a differential inequality. We write it in compact form (see \Cref{cor:differential-inequality}) as
\begin{equation}\label{eq:diffineq}
    D^+F(e,t)\le 2\Delta F(e,t),
\end{equation}
where $D^+$ is the upper derivative in $t$, defined in \Cref{sec:2}, and $\Delta$ is the adjacency operator of the line graph of $G$ (\Cref{def:line-graph}): two edges are adjacent when they share a vertex.
A comparison argument gives
\[
F(e,t)\le e^{2t\Delta}F(e,0).
\]

Since $A$ is initially supported at $x$, the initial values $F(e,0)$ vanish unless $e$ contains $x$. The original commutator is therefore bounded by matrix entries of $e^{2t\Delta}$ between edges containing $x$ and edges containing $y$. This reduces the many-body propagation problem to an off-diagonal estimate for $e^{2t\Delta}$.

One standard way to estimate $e^{2t\Delta}$ is to expand the exponential in powers of $\Delta$. The entries of $\Delta^n$ count walks of length $n$ in the line graph, so this leads to a path-counting argument. This approach is natural and has been used extensively in the literature on Lieb--Robinson bounds \cite{Hastings04, HastingsKoma, NachtergaeleSims06}, where one counts sequences of successively overlapping interaction terms. Our approach instead uses the band structure of $\Delta$: its matrix entries vanish unless the corresponding interaction edges overlap. The Combes--Thomas estimate converts this locality into exponential off-diagonal decay for the resolvent
\[
(z-\Delta)^{-1}.
\]
A contour-integral representation of $e^{2t\Delta}$ then transfers the resolvent estimate to the matrix exponential. Optimizing the contour gives the final light-cone bound and the superexponential decay outside the cone.

The differential inequality and its reduction to a matrix exponential are closely related to existing approaches to Lieb--Robinson bounds. Papers using this technique include the aforementioned \cite{LiebRobinson, Hastings04, HastingsKoma, NachtergaeleSims06, NachtergaeleSims}. Of these, \cite{Hastings04, HastingsKoma, NachtergaeleSims06, NachtergaeleSims} then prove the Lieb--Robinson bound by the path-counting argument described above, while \cite{LiebRobinson} estimates the resulting kernel by a Fourier contour integral, relying on translation invariance. For a good pedagogical survey of such methods, one can refer to \cite{Chen}. The Combes--Thomas estimate originates in spectral theory \cite{CombesThomas} and has been used in various single-particle problems, see e.g. \cite[Chapter~10]{AizenmanWarzel}, \cite{Fischer00}, and \cite[Proposition~2.8]{KongEtAl23}. The present method combines these ideas.

\subsection{Some open questions} We suggest here a few open problems to illustrate the flexibility of our new approach.

The Combes--Thomas estimate provided in \cite{AizenmanWarzel} is stated for an \emph{arbitrary} discrete metric space
$(\Lambda,d)$ and an arbitrary self-adjoint operator on $\ell^2(\Lambda)$:
neither the graph structure of $G$ nor the finite range of the interaction
enters its proof. In our argument the finite range enters through the differential inequality \eqref{eq:diffineq} and through the bounds $\|\Delta\|\le 2\DD$ and $S_\mu\le 2\DD(e^\mu-1)$ on the line-graph operator. One
therefore expects the proof to go through once $\Delta$ is replaced by its analogue for a decaying interaction in the sense of \cite[Assumption~2.4]{HastingsKoma},
\begin{equation}
  H=\sum_{Z\subseteq\Lambda}H_Z,
  \qquad
  \sup_{x\in\Lambda}\sum_{Z\ni x}\lVert H_Z\rVert\,|Z|\,e^{\nu\operatorname{diam}Z} \le C
  \qquad \nu>0,
\end{equation} where $C> 0$ is a constant independent of $|\Lambda|$.

One could also treat Hamiltonians with algebraic decay, e.g. with interaction terms $H_{(x,y)}$ that decay like $|x-y|^{-\alpha}$ for $\alpha > d$, where $d$ is the geometric dimension of $\Lambda$ (the exponent that controls the volume growth of balls). Lieb--Robinson bounds for such interactions, with light cones depending on $\alpha$, are known \cite{FossFeig15, ChenLucas19, KuwaharaSaito20, TranEtAl21}. While the Combes--Thomas inequality does not directly apply to this case, improvements such as \cite{Sun11} do. The strategy described in our work would then likely yield an algebraic Lieb--Robinson bound for such systems.

It would likewise be interesting to look at graphs whose degree grows with $|\Lambda|$ instead of being merely bounded. Some limitations are clearly needed: for the complete graph, every site directly influences all other sites, so no meaningful Lieb--Robinson bound can hold. One avenue is to require that the average degree does not grow too fast with $|\Lambda|$. Sparse random graphs provide a natural model. For any connected subgraph of $G(n,c/n)$, where \[
    G = (V, E), \quad\text{with}\quad V = \{1, \dots, n\}, \quad\text{and}\quad \mathbb{P}(i \sim j) = \frac{c}{n}\quad \text{for } \ i \ne j \in V
\] the maximal degree is, with high probability, of order $\log n / \log\log n$ \cite[Theorem~3.4]{FriezeKaronski15} (see also \cite{FriezeRaut24}), and \Cref{thm:main} gives a velocity growing like: \[
    v = 32\cdot \frac{\log n}{\log\log n}.
\] However, the graph is locally tree-like with average degree $c$ so one might expect a size-independent velocity for a typical pair of sites. It is unclear whether the deviation from a constant velocity is a mathematical artifact or whether infrequent hubs really do accelerate propagation between typical sites. The proof of our main result, \Cref{thm:main}, uses the degree of the graph only through the bounds $\|\Delta\|\le 2\DD$ and $S_\mu\le 2\DD(e^\mu-1)$ and through $\deg(x)$ in \Cref{thm:matrix-exponential}, and treats the Combes--Thomas estimate as a black box. Providing slight adjustments to a weight or distance function on the graph based on maximal degree, then carrying the argument of the paper through, may be easier than other means of eluding the issue of maximal degree. 

\subsection*{Acknowledgments}
This paper grew out of a research project carried out as part of the Program for Research in Mathematical Exploration (PRIME) in the Department of Mathematics at the University of Washington. First and foremost, we are grateful to our mentor, Professor Alexis Drouot, for proposing this project and its spectral approach, for guidance throughout, and for detailed comments on every draft. We thank Professor Rekha Thomas and Sarah Garner, Director of Student Services, for organizing the program. Finally, we found the historical survey by Chen, Lucas, and Yin \cite{Chen} especially valuable as a guide to the methods for proving Lieb--Robinson bounds. This project benefited from NSF support through the grants DMS\nobreakdash-2439949 and DMS\nobreakdash-2054589.

\subsection*{Use of artificial intelligence} We used Claude Fable 5.1 in editing stages of this paper: to check proofs and citations against the sources, to flag typographical and consistency errors, to typeset figures and bibliography entries, and to assist in improving the wording. The mathematical ideas, the proofs, and the final text are the authors' own.

\section{Proof of the Differential Inequality}\label{sec:2}
Throughout this section, $G=(\Lambda,E)$ and $H=\sum_{e\in E}H_e$ are as in \Cref{sec:setup}, and $A$ is an operator on $\HH$ with $\|A\|\le 1$. For an operator $X$ on $\HH$ we write
\begin{equation}
    f_X(t)=\|[X,A(t)]\|,
\end{equation}
and for an edge $e\in E$ we recall from \Cref{sec:proof-strategy} the edge function
\begin{equation}
    F(e,t)=\sup_{\substack{\supp(X)\subseteq e,\\ \|X\|\le 1}} f_X(t).
\end{equation}
The quantity $F(e,t)$ measures the influence of $A(t)$ on the edge $e$: by \Cref{thm:approximate-support}, $F(e,t)\le\varepsilon$ means that $A(t)$ is within $\varepsilon$ of an operator supported in $\Lambda\setminus e$. The goal of this section is a differential inequality for $F(e,t)$ in $t$ (\Cref{lem:differential-inequality} and \Cref{cor:differential-inequality}). Since $F(e,\cdot)$ need not be differentiable, we use its \emph{upper derivative}
\begin{equation}
    D^+F(e,t) \defeq \limsup_{h\to 0^+}\frac{F(e,t+h)-F(e,t)}{h},
\end{equation}
which is finite because $F(e,\cdot)$ is Lipschitz, as we show first; for the basic properties of such Dini derivatives see \cite{HT06}.
\begin{lemma}
    \leavevmode\phantomsection\label{thm:Lipschitz}For every $e\in E$ and all $s,t\in\R$,
    \[
        |F(e,t)-F(e,s)|\le 2\,\|[H,A]\|\,|t-s| .
    \]
    In particular, $D^+F(e,t)$ is finite for every $t\in\R$.
\end{lemma}
\begin{proof}
        Fix $X$ with $\supp(X)\subseteq e$ and $\|X\|\le 1$. For $s\le t$,
    \begin{equation}\label{F-lipschitz}
        \begin{aligned}
            |f_X(t) - f_X(s)| &= \big|\|[X,A(t)]\| - \|[X,A(s)]\|\big| \\
            & \le 2\|X\|\|A(t)-A(s)\| \le 2\int_s^t \|A'(u)\|du \le 2\|[iH,A]\|(t-s).
        \end{aligned}
    \end{equation}
    Here the first inequality uses the reverse triangle inequality and the linearity of the commutator, $\bigl|\|[X,A(t)]\|-\|[X,A(s)]\|\bigr|\le\|[X,A(t)-A(s)]\|\le 2\|X\|\,\|A(t)-A(s)\|$; the second is the fundamental theorem of calculus with $\|X\|\le 1$; the third uses $\|A'(u)\|=\|e^{iHu}[iH,A]e^{-iHu}\|=\|[iH,A]\|=\|[H,A]\|$ by unitary invariance. Since $|\sup_X f_X(t)-\sup_X f_X(s)|\le\sup_X|f_X(t)-f_X(s)|$, taking the supremum over $X$ in \eqref{F-lipschitz} gives $|F(e,t)-F(e,s)|\le 2\|[H,A]\|\,(t-s)$ for $s\le t$, and the case $s>t$ follows by symmetry. The difference quotients of $F(e,\cdot)$ are therefore bounded by $2\|[H,A]\|$, so $D^+F(e,t)$ is finite.
\end{proof}
We are now prepared to establish the following differential inequality on $F(e,t)$, which goes back to \cite{LiebRobinson} (see also \cite{Hastings04, HastingsKoma, NachtergaeleSims}); we include a proof for completeness.

\begin{lemma}
    \leavevmode\phantomsection\label{lem:differential-inequality}For every $t\in \R$,
     \begin{equation}
        D^+F(e,t) \le 2\sum_{e':e'\cap e\ne \emptyset} F(e',t).
    \end{equation}
\end{lemma}
\begin{proof}
    We proceed in two steps. First, for a fixed operator $X$, we conjugate $[X,A(t+h)]$ into $[X(-h),A(t)]$ and write the latter as a unitary conjugate of $[X,A(t)]$, plus a first-order term involving the edges neighboring $e$, plus remainders of order $h^2$. Second, taking norms gives an inequality for the increment $\|[X,A(t+h)]\|-\|[X,A(t)]\|$ whose right-hand side does not depend on $X$, so we may take the supremum over $X$; letting $h\to 0^+$ then yields the upper derivative.

    \medskip\noindent\textit{Expansion of $[X,A(t+h)]$ for fixed $X$.} Fix $h>0$ and an operator $X$ with $\|X\|\le 1$ and $\supp(X)\subseteq e$. Since the norm is invariant under unitary conjugation,
    \begin{equation}
        \|[X,A(t+h)]\|=\|e^{-ih H}[X,A(t+h)]e^{ih H}\|=\|[X(-h),A(t)]\|,
        \qquad X(u)\defeq e^{iuH}Xe^{-iuH}.
    \end{equation}
    Expanding $X(u)$ by its Taylor series about $u=0$ gives
    \begin{equation}\label{eq:X-expansion}
        X(-h)= X- ih [H,X]+R_1,
        \qquad
        R_1\defeq \int_0^{-h}(-h-u)X''(u)\,du .
    \end{equation}
    The remainder $R_1$ is of order $h^2$, with a constant independent of $X$:
    \begin{equation}\label{remainder-bound}
        \begin{aligned}
            \|R_1\| &\le \int_{-h}^{0}|-h-u|\,\|X''(u)\|\,du\\
            &\le\frac{h^2}{2}\sup_{u\in [-h,0]}\|X''(u)\|=\frac{h^2}{2}\|[H,[H,X]]\|\le 2\|H\|^2h^2 ,
        \end{aligned}
    \end{equation}
    where the equality uses $X''(u)=-e^{iHu}[H,[H,X]]e^{-iHu}$ and the last inequality uses $\|X\|\le 1$.

    We next examine the commutator $[H,X]$ in the first-order term. Since $\supp(X)\subseteq e$ and operators with disjoint supports commute, only the edges meeting $e$ contribute to it. We henceforth write
    \begin{equation}
        \mathcal N \defeq \{e'\in E : e'\cap e\neq\emptyset\}, \qquad H_{\mathcal N}\defeq\sum_{e'\in\mathcal N}H_{e'},
    \end{equation}
    so that
    \begin{equation}
        [H,X]=\Big[\sum_{e'\in E}H_{e'},X\Big]=\sum_{e'\in E}[H_{e'},X]=\sum_{e'\in\mathcal N}[H_{e'},X] .
    \end{equation}
    Substituting this into \eqref{eq:X-expansion}, we may rewrite the commutator $[X(-h),A(t)]$ as
    \begin{equation}\label{eq:recast}
        [X(-h),A(t)] = [X,A(t)]-ih\sum_{e'\in\mathcal N}[[H_{e'},X],A(t)]+[R_1,A(t)] .
    \end{equation}
    Applying the Jacobi identity $[[H_{e'},X],A(t)]=[H_{e'},[X,A(t)]]-[X,[H_{e'},A(t)]]$ to each term of the sum yields
    \begin{equation}\label{eq:jacobi-split}
        -ih \sum_{e'\in\mathcal N} [[H_{e'},X],A(t)] = -ih\, [H_{\mathcal N},[X,A(t)]] + ih \sum_{e'\in\mathcal N} [X,[H_{e'},A(t)]] .
    \end{equation}
    Now observe that the first term on the right is the first-order term of the unitary conjugation $u\mapsto e^{-i u H_{\mathcal N}}[X,A(t)] e^{i u H_{\mathcal N}}$; expanding it about $u=0$ as in \eqref{eq:X-expansion} gives
    \begin{equation}\label{eq:conjugation-expansion}
        \begin{aligned}
        e^{-ih H_{\mathcal N}}[X,A(t)] e^{ih H_{\mathcal N}}&=[X,A(t)]-ih[H_{\mathcal N},[X,A(t)]]+R_2,\\
        R_2&\defeq \int_0^{h}(h-u)\frac{d^{2}}{du^{2}}\Big(e^{-iuH_{\mathcal N}}[X,A(t)]e^{iuH_{\mathcal N}}\Big)du .
        \end{aligned}
    \end{equation}
    Substituting \eqref{eq:jacobi-split} into \eqref{eq:recast} and using \eqref{eq:conjugation-expansion} to replace the first two terms produces
    \begin{equation}\label{eq:combined}
        [X(-h), A(t)] = e^{-ih H_{\mathcal N}}[X, A(t)]e^{ih H_{\mathcal N}} + ih \sum_{e' \in \mathcal N} [X, [H_{e'}, A(t)]] + [R_1, A(t)] - R_2 .
    \end{equation}

    \medskip\noindent\textit{Uniform bound and the supremum.} Note that each term in the middle sum of \eqref{eq:combined} satisfies
    \begin{equation}
        \|[X, [H_{e'}, A(t)]]\|\le 2\|X\|\|[H_{e'}, A(t)]\|\le 2\|[H_{e'}, A(t)]\|,
    \end{equation}
    and that, by \eqref{remainder-bound} and the same argument for $R_2$ (using $\|[X,A(t)]\|\le 2$),
    \begin{equation}
        \|[R_1, A(t)]\| \le 2\|R_1\| \le 4\|H\|^2h^2,
        \qquad
        \|R_2\| \le \frac{h^2}{2}\,\bigl\|[H_{\mathcal N}, [H_{\mathcal N}, [X, A(t)]]]\bigr\| \le 4\|H_{\mathcal N}\|^2h^2 .
    \end{equation}
    Taking norms in \eqref{eq:combined}, applying the triangle inequality, and using once more that the norm is invariant under unitary conjugation, it follows that
    \begin{equation}\label{eq:1i}
        \|[X, A(t+h)]\| \le \|[X, A(t)]\| + 2h \sum_{e' \in \mathcal N} \|[H_{e'}, A(t)]\| + Ch^2,
        \qquad C\defeq 4\|H\|^2+4\|H_{\mathcal N}\|^2 .
    \end{equation}
    We now establish a bound on the right-hand side of the inequality independent of $X$. Trivially, $\|[X, A(t)]\|\le F(e,t)$. Further, for each term with $H_{e'}\ne 0$ in the sum, ${H_{e'}}/{\|H_{e'}\|}$ has norm bounded by $1$ and support contained in $e'$, so
    \begin{equation}
        \|[H_{e'},A(t)]\|=\|H_{e'}\|\,\|[H_{e'}/\|H_{e'}\|,A(t)]\|\le \|H_{e'}\|F(e',t)\le F(e',t) .
    \end{equation}
    Since the right-hand side no longer depends on $X$, we may take the supremum over $X$ satisfying $\|X\|\le 1$, $\supp(X)\subseteq e$ to produce
    \begin{equation}\label{eq:1j}
        F(e, t + h) \le F(e, t) + 2h \sum_{e' \in \mathcal N} F(e', t) + Ch^2 .
    \end{equation}
    Subtracting $F(e,t)$ from both sides, dividing by $h>0$ and letting $h\to 0^+$, we obtain
    \begin{equation}\label{eq:F-differential-inequality}
        D^+F(e,t) \le 2\sum_{e'\in\mathcal N} F(e',t),
    \end{equation}
    where the upper derivative is the appropriate notion of derivative, as explained at the beginning of this section, and is finite by \Cref{thm:Lipschitz}.
\end{proof}

The right-hand side of the inequality in \Cref{lem:differential-inequality} couples an edge only to the edges with which it shares a vertex. This is the adjacency relation of the line graph of $G$, which we now make explicit (see \Cref{fig:line-graph}).

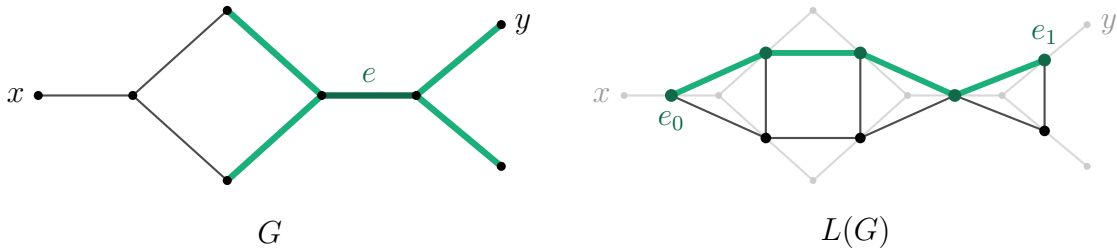
\begin{figure}[ht]
\centering
\begin{tikzpicture}[x=1.25cm,y=1.25cm]
    \samplegraph
    \draw[black!70,line width=0.8pt] \samplegraphedges;
    \draw[figaccent,line width=2.4pt,line cap=round] (p2)--(p4) (p3)--(p4) (p5)--(y) (p5)--(r);
    \draw[figaccent!60!black,line width=2.4pt,line cap=round] (p4)--(p5);
    \foreach \v in {x,p1,p2,p3,p4,p5,y,r} \fill (\v) circle (1.7pt);
    \node[figaccent!60!black,above=1pt] at ($(p4)!.5!(p5)$) {$e$};
    \node[left=1pt] at (x) {$x$};
    \node[right=1pt] at (y) {$y$};
    \node at (2.45,-1.45) {$G$};
    \begin{scope}[shift={(6.2,0)}]
    \samplegraph
    \draw[black!15,line width=0.8pt] \samplegraphedges;
    \foreach \v in {x,p1,p2,p3,p4,p5,y,r} \fill[black!20] (\v) circle (1.3pt);
    \node[black!40,left=1pt] at (x) {$x$};
    \node[black!40,right=1pt] at (y) {$y$};
    \foreach \a/\b [count=\i] in {x/p1,p1/p2,p1/p3,p2/p4,p3/p4,p4/p5,p5/y,p5/r}
        \coordinate (E\i) at ($(\a)!.5!(\b)$);
    \draw[black!70,line width=0.8pt] (E1)--(E2) (E1)--(E3) (E2)--(E3) (E2)--(E4) (E3)--(E5)
        (E4)--(E5) (E4)--(E6) (E5)--(E6) (E6)--(E7) (E6)--(E8) (E7)--(E8);
    \draw[figaccent,line width=2.2pt,line cap=round,line join=round] (E1)--(E2)--(E4)--(E6)--(E7);
    \foreach \i in {3,5,8} \fill (E\i) circle (2pt);
    \foreach \i in {1,2,4,6,7} \fill[figaccent!60!black] (E\i) circle (2.4pt);
    \node[figaccent!60!black,below=2pt] at (E1) {$e_0$};
    \node[figaccent!60!black,above=2pt] at (E7) {$e_1$};
    \node at (2.45,-1.45) {$L(G)$};
    \end{scope}
\end{tikzpicture}%
\caption{Left: a graph $G$ with $\DD=3$; highlighted, the set $\mathcal N$ of edges sharing a vertex with $e$ (dark). Right: its line graph $L(G)$, each vertex drawn at the midpoint of the corresponding edge of $G$ (gray). A shortest path from $e_0\ni x$ to $e_1\ni y$ has length $d(e_0,e_1)=4=d(x,y)-1$.}
\label{fig:line-graph}
\end{figure}

\begin{definition}[Line graph and its adjacency operator]\leavevmode\phantomsection\label{def:line-graph}%
Two edges $e,e'\in E$ are \emph{adjacent}, written $e\sim e'$, if they share a vertex, that is, if $e\cap e'\neq\emptyset$; in particular $e\sim e$. The \emph{line graph} $L(G)$ of $G$ is the set $E$ equipped with this adjacency relation. We write $\deg_{L(G)}(e)$ for the number of edges other than $e$ that are adjacent to $e$, and $d(e,e')$ for the \emph{distance} between $e$ and $e'$ in $L(G)$, that is, the smallest $n\ge 0$ for which there are edges $e=e_0,e_1,\dots,e_n=e'$ with $e_{j-1}\sim e_j$ for every $j$. The \emph{adjacency operator} $\Delta$ of $L(G)$ acts on functions $g\colon E\to\R$ by
    \begin{equation}\label{eq:line-graph}
      (\Delta g)(e) = \sum_{e'\sim e} g(e') = \sum_{e'\,:\,e\cap e'\neq\emptyset} g(e').
    \end{equation}
\end{definition}
The matrix of $\Delta$ is real and symmetric, with entries $\Delta(e,e')=1$ if $e\sim e'$ and $0$ otherwise; in particular, all entries of $\Delta$ are nonnegative, and the same formula defines a self-adjoint operator on $\ell^2(E)$. Since $G$ is connected, $d$ is a metric on $E$. With this notation, \Cref{lem:differential-inequality} takes the following compact form.

\begin{corollary}\label{cor:differential-inequality}
For every $t\in\R$ and every $e\in E$,
    \begin{equation}
         D^+F(e, t) \leq \big(2\Delta F(\cdot,t)\big) (e),
    \end{equation}
where $\Delta$ acts on the edge variable of $F(\cdot,t)$.
\end{corollary}
\begin{proof}
The sum in \Cref{lem:differential-inequality} is over the edges $e'\sim e$, so by \eqref{eq:line-graph} its right-hand side is $2(\Delta F(\cdot,t))(e)$.
\end{proof}

  \begin{theorem}
    \leavevmode\phantomsection\label{thm:matrix-exponential}Let $x,y\in\Lambda$. Suppose $A$ and $B$ are operators satisfying
    \[
    \|A\|\le 1,\qquad \operatorname{supp}(A)\subseteq \{x\},
    \quad \text{and} \quad
    \|B\|\le 1,\qquad \operatorname{supp}(B)\subseteq \{y\}.
    \]
    Define
    \[
    f(t):=\|[B,A(t)]\|.
    \]
    Then, for every $t\ge 0$ and every edge $e_1\in E$ containing $y$,
    \[
    f(t)=\|[A(t),B]\|
    \le F(e_1,t)
    \le
    2\sum_{\substack{e_0\in E\\ x\in e_0}}
    \left(e^{2t\Delta}\right)(e_1,e_0)
    \le 2\deg(x)\max_{\substack{e_0\in E\\ x\in e_0}}\left(e^{2t\Delta}\right)(e_1,e_0).
    \]
  \end{theorem}
 \begin{proof}
    By \Cref{thm:Lipschitz}, for each $e\in E$, the function $F(e,\cdot)$ is
    Lipschitz, hence absolutely continuous and differentiable almost everywhere.
    At every point $s$ where $F(e,\cdot)$ is differentiable, its ordinary
    derivative agrees with its upper derivative \cite[p.~34]{HT06}. Therefore, by \Cref{cor:differential-inequality},
    \[
    F'(e,s)\le (2\Delta F(\cdot,s))(e)
    \]
    for almost every $s$.
    
    Fix $t\ge 0$. For $0\le s\le t$, define
    \[
    G(s):=e^{2(t-s)\Delta}F(\cdot,s).
    \]
    Since each component of $F(\cdot,s)$ is absolutely continuous and
    $s\mapsto e^{2(t-s)\Delta}$ is smooth, $G$ is absolutely continuous.
    Thus, for almost every $s\in[0,t]$,
    \[
    \begin{aligned}
    G'(s)
    &=
    -2\Delta e^{2(t-s)\Delta}F(\cdot,s)
    +
    e^{2(t-s)\Delta}F'(\cdot,s)\\
    &=
    e^{2(t-s)\Delta}
    \left(F'(\cdot,s)-2\Delta F(\cdot,s)\right),
    \end{aligned}
    \]
    where we used that $\Delta$ commutes with $e^{2(t-s)\Delta}$.
    
    By \Cref{cor:differential-inequality},
    \[
    F'(\cdot,s)-2\Delta F(\cdot,s)\le 0
    \]
    componentwise for almost every $s$. Moreover, $\Delta$ has nonnegative
    matrix entries, so
    \[
    e^{2(t-s)\Delta}
    =
    \sum_{n=0}^{\infty}
    \frac{(2(t-s)\Delta)^n}{n!}
    \]
    also has nonnegative matrix entries for $0\le s\le t$. Hence
    \[
    G'(s)\le 0
    \]
    componentwise for almost every $s$. Since $G$ is absolutely continuous,
    \[
    G(t)-G(0)=\int_0^t G'(s)\,ds\le 0.
    \]
    Therefore,
    \[
    F(\cdot,t)\le e^{2t\Delta}F(\cdot,0).
    \]
    
    Taking the component corresponding to an edge $e_1\in E$ containing $y$,
    we obtain
    \[
    F(e_1,t)
    \le
    \sum_{e_0\in E}
    \left(e^{2t\Delta}\right)(e_1,e_0)F(e_0,0).
    \]
    
    Since $A$ is supported at $x$, if $x\notin e_0$, then every operator
    supported on $e_0$ has support disjoint from that of $A$, and hence
    commutes with $A$. Thus
    \[
    F(e_0,0)=0
    \qquad\text{whenever }x\notin e_0.
    \]
    Therefore,
    \[
    F(e_1,t)
    \le
    \sum_{\substack{e_0\in E\\x\in e_0}}
    \left(e^{2t\Delta}\right)(e_1,e_0)F(e_0,0).
    \]
    For every $e_0$ containing $x$,
    \[
    F(e_0,0)
    =
    \sup_{\substack{\operatorname{supp}(X)\subseteq e_0\\ \|X\|\le1}}
    \|[X,A]\|
    \le 2\|A\|
    \le 2,
    \]
    so
    \[
    F(e_1,t)
    \le
    2\sum_{\substack{e_0\in E\\x\in e_0}}
    \left(e^{2t\Delta}\right)(e_1,e_0).
    \]
    
    Finally, since $B$ is supported at $y$ and $y\in e_1$, while $\|B\|\le1$,
    the operator $B$ is admissible in the definition of $F(e_1,t)$. Hence
    \[
    f(t)
    =
    \|[B,A(t)]\|
    =
    \|[A(t),B]\|
    \le
    F(e_1,t).
    \]
    Since there are exactly $\deg(x)$ edges containing $x$, and all entries
    of $e^{2t\Delta}$ are nonnegative,
    \[
    \sum_{\substack{e_0\in E\\x\in e_0}}
    \left(e^{2t\Delta}\right)(e_1,e_0)
    \le
    \deg(x)
    \max_{\substack{e_0\in E\\x\in e_0}}
    \left(e^{2t\Delta}\right)(e_1,e_0).
    \]
    Combining the preceding inequalities gives the asserted chain.
    \end{proof}
\section{Combes--Thomas Bound and Proof of the Main Theorem}\label{sec:3}

\Cref{thm:matrix-exponential} gives the bound:
\begin{equation}\label{eq:qss-edge-recursion}
 \|[A(t), B]\| \le 2\deg(x) \cdot \bigl(e^{2t\Delta}\bigr)(e_1,e_0),
\end{equation} 
where $A$ is supported at a site $x$ with $\|A\|\le 1$, $B$ is supported at a site $y$ with $\|B\|\le 1$, $t\ge 0$, $e_1\in E$ is any edge containing $y$, and $e_0\in E$ is an edge containing $x$ at which $(e^{2t\Delta})(e_1,\cdot)$ is largest; this edge depends on $t$. 

In this section, we want to estimate the entries of $e^{2t\Delta}$ when the corresponding edges are far apart with respect to the line graph distance. Earlier works \cite{Hastings04, HastingsKoma, NachtergaeleSims06, NachtergaeleSims} have performed this step by expanding $e^{2t\Delta}$ using path counting. Instead, we use Cauchy's integral formula to express the matrix exponential in terms of resolvents $(z-\Delta)^{-1}$:
\begin{equation}
    e^{2t\Delta} = \frac{1}{2\pi i}\int_{\Gamma} e^{tz}(z- 2\Delta)^{-1}dz
\end{equation} where $\Gamma \subseteq \mathbb{C}$ is the boundary of an open domain containing the full spectrum of $2\Delta$. We will consider here, for $\rho>0$, the stadium-shaped contour $\Gamma_\rho$ (\Cref{fig:stadium}).  

\begin{figure}[b]
\centering
\begin{tikzpicture}[scale=1.6,>={Stealth[length=8pt,width=6pt]}]
    \def\LA{2}\def\rh{1.7}
    \fill[figaccent!8]
        (-\LA,-\rh) -- (\LA,-\rh)
        arc[start angle=-90,end angle=90,radius=\rh]
        -- (-\LA,\rh)
        arc[start angle=90,end angle=270,radius=\rh] -- cycle;
    \draw[->,black!40] (-5.0,0) -- (5.0,0) node[anchor=south east,black!55,inner sep=2pt,xshift=-7pt,yshift=3pt]{$\operatorname{Re}z$};
    \draw[->,black!40] (0,-2.2) -- (0,2.75) node[anchor=north east,black!55,inner sep=2pt,xshift=-4pt,yshift=-1pt]{$\operatorname{Im}z$};
    \draw[figaccent,line width=1.1pt]
        (-\LA,-\rh) -- (\LA,-\rh)
        arc[start angle=-90,end angle=90,radius=\rh]
        -- (-\LA,\rh)
        arc[start angle=90,end angle=270,radius=\rh] -- cycle;
    \node[figaccent,anchor=south east,inner sep=4pt] at ($(-\LA,0)+(130:\rh)$) {$\Gamma_\rho$};
    \node[figaccent!60!black] at (-1.3,0.9) {$D_\rho$};
    \fill[figaccent!60!black] (1.0,\rh) circle (1.6pt);
    \node[figaccent!60!black,above] at (1.0,\rh) {$z$};
    \draw[densely dotted,black!70] (1.0,\rh) -- (1.0,0);
    \node[black!70,anchor=east,inner sep=2pt] at (1.0,0.85) {$\rho$};
    \draw[dashed,black!70] (1.0,\rh) -- (1.3,0);
    \node[black!70,anchor=west,inner sep=2pt] at (1.2,0.95) {$\delta$};
    \draw[-{Stealth[length=11pt,width=9pt]},figaccent,line width=1.1pt] (-0.8,\rh) -- (-1.2,\rh);
    \draw[-{Stealth[length=11pt,width=9pt]},figaccent,line width=1.1pt] (0.8,-\rh) -- (1.2,-\rh);
    \draw[black,line width=2pt] (-\LA,0) -- (\LA,0);
    \foreach \x in {-1.7,-1.2,-0.4,0.5,1.3,2} \fill (\x,0) circle (1.2pt);
    \node[black!70,below=3pt] at (-0.8,0) {$\operatorname{spec}(M)$};
    \node[black!70,below=4pt] at (-\LA,0) {$\mathllap{-}\norm{M}$};
    \node[black!70,below=4pt] at (\LA,0) {$\norm{M}$};
    \draw[<->,black!60] (-\LA,0) -- ++(145:\rh-0.03);
    \node[black!70] at ($(-\LA,0)+(145:0.5*\rh)+(0.12,0.2)$) {$\rho$};
    \draw[<->,black!60] (\LA,0) -- ++(35:\rh-0.03);
    \node[black!70] at ($(\LA,0)+(35:0.5*\rh)+(-0.12,0.2)$) {$\rho$};
    \draw[black!70] (-0.07,\rh) -- (0.07,\rh);
    \draw[black!70] (-0.07,-\rh) -- (0.07,-\rh);
    \node[black!70,anchor=south east,inner sep=2pt] at (-0.03,\rh+0.03) {$\rho$};
    \node[black!70,anchor=north east,inner sep=2pt] at (-0.03,-\rh-0.03) {$-\rho$};
    \draw[black!70] (\LA+\rh,0.08) -- (\LA+\rh,-0.08);
    \node[black!70,anchor=north west] at (\LA+\rh-0.05,-0.1) {$\norm{M}+\rho$};
\end{tikzpicture}
\caption{The stadium contour $\Gamma_\rho$.}
\label{fig:stadium}
\end{figure}
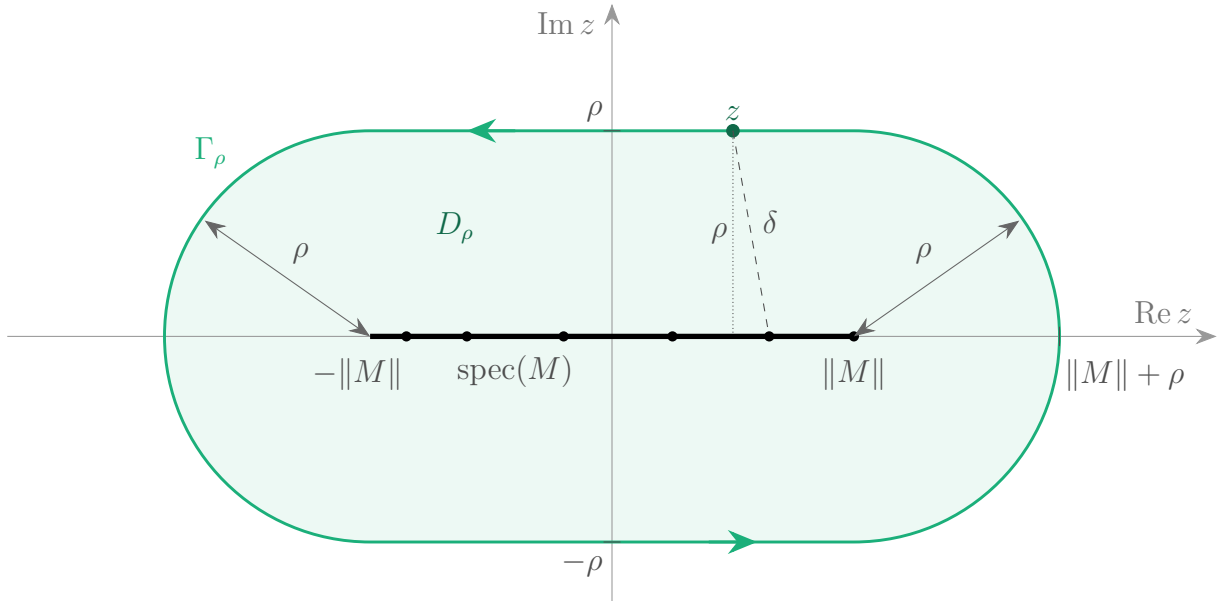

Off-diagonal decay of resolvents of band matrices was first proved by Combes and Thomas in \cite{CombesThomas}. We use a finite-space formulation from \cite[Theorem 10.5]{AizenmanWarzel}.

\begin{theorem}[Combes--Thomas estimate]\label{thm:ct}
Let $(X,d)$ be a finite metric space and let $A$ be a self-adjoint operator on $\ell^2(X)$. For $\mu>0$, set
\[
S_\mu \defeq \sup_{x\in X}\ \sum_{y\in X} |A(x,y)|\bigl(e^{\mu\,d(x,y)}-1\bigr).
\]
Then, for every $z\notin\operatorname{spec}(A)$ and every $\mu>0$ with $S_\mu<\delta$, where $\delta\defeq\operatorname{dist}\bigl(z,\operatorname{spec}(A)\bigr)$,
\begin{equation}\label{eq:CT-exponential}
\bigl|(z-A)^{-1}(x,y)\bigr| \le \frac{e^{-\mu\,d(x,y)}}{\delta-S_\mu}
\qquad\text{for all } x,y\in X .
\end{equation}
\end{theorem}

According to \Cref{thm:ct}, if $z$ is far enough from the spectrum, in the sense $\delta>S_\mu$, the resolvent matrix element between distant points decays exponentially at rate $\mu$. Values $z \in \Gamma_\rho$ stay at distance at least $\rho$ from the spectrum, and this allows us to apply the Combes--Thomas estimate to $A$ once $\rho>S_\mu$, with $S_\mu$ as in \Cref{thm:ct}, which yields decay of the off-diagonal entries of $e^{tA}$.

\begin{lemma}\label{lem:exponential-decay}
    Let $M$ be a self-adjoint operator on $\ell^2(X)$, and let $\mu,\rho>0$ be such that $S_\mu<\rho$. For all $t>0$ and all $x,y\in X$,
    \begin{equation}\label{eq:improved-bound}
    \bigl|e^{tM}(x,y)\bigr|
    \le \frac{2\norm{M}+\pi\rho}{\pi(\rho-S_\mu)}\;e^{(\norm{M}+\rho)\,t-\mu\,d(x,y)} .
\end{equation}
\end{lemma}
\begin{proof}
Let $D_\rho\subseteq\C$ denote the interior of the stadium-shaped contour $\Gamma_\rho$ defined as the set of points at distance $\rho$ from $[-\| M \|, \|M\|]$, oriented positively (see \Cref{fig:stadium}).

Since $M$ is self-adjoint,
\begin{equation}\label{eq:1h}
\operatorname{spec}(M)\subseteq[-\norm{M},\norm{M}].
\end{equation}
By the definition of $\Gamma_\rho$,
\[
[-\norm{M},\norm{M}]\subseteq D_\rho\subseteq\C.
\]

The domain $D_\rho$ is bounded and its boundary $\Gamma_\rho$ is piecewise-smooth. Hence, for every $\lambda\in\operatorname{spec}(M)$, by the Cauchy Integral Formula:
\[
e^{t\lambda} = \frac{1}{2\pi i}\oint_{\Gamma_\rho} \frac{e^{tz}}{z-\lambda}\,dz.
\]
Tracing this argument for each eigenvector of $M$, and using that $M$ is diagonalizable, we obtain:
\[
e^{tM} \defeq \sum_{n = 0}^\infty \frac{(tM)^n}{n!}= \frac{1}{2\pi i}\oint_{\Gamma_\rho} e^{tz}\,(z-M)^{-1}\,dz.
\]
Taking the $(x,y)$ matrix entry and applying the triangle inequality,
\[
\bigl|e^{tM}(x,y)\bigr| \le \frac{1}{2\pi}\oint_{\Gamma_\rho} \bigl|e^{tz}\bigr|\,\bigl|(z-M)^{-1}(x,y)\bigr|\,|dz|.
\]

For $z\in\Gamma_\rho$, let 
$
\delta \defeq \operatorname{dist}(z,\operatorname{spec}(M)).
$
Since $\Gamma_\rho$ lies at distance $\rho$ from $[-\norm{M},\norm{M}]$, combining this with \eqref{eq:1h}
we have
$
\delta\ge\rho.
$
Since $S_\mu<\rho$, it follows that
$
S_\mu<\rho\le\delta.
$
Thus, by the Combes--Thomas estimate (\Cref{thm:ct}),
\[
\bigl|(z-M)^{-1}(x,y)\bigr| \le \frac{1}{\delta-S_\mu}e^{-\mu\,d(x,y)} \le \frac{1}{\rho-S_\mu}e^{-\mu\,d(x,y)}.
\]
Also, for $z\in\Gamma_\rho$,
$
\operatorname{Re}z\le\norm{M}+\rho,
$
and hence
\[
\bigl|e^{tz}\bigr| = e^{t\operatorname{Re}z} \le e^{(\norm{M}+\rho)t}.
\]
Finally, $|\Gamma_\rho| = 4\norm{M}+2\pi\rho$ (see \Cref{fig:stadium}).
Substituting these estimates into the contour integral,
\[
\bigl|e^{tM}(x,y)\bigr| \le \frac{1}{2\pi}\bigl(4\norm{M}+2\pi\rho\bigr)\,e^{(\norm{M}+\rho)t}\,\frac{1}{\rho-S_\mu}\,e^{-\mu\,d(x,y)}.
\]
Thus,
\[
\bigl|e^{tM}(x,y)\bigr| \le \frac{2\norm{M}+\pi\rho}{\pi(\rho-S_\mu)}\,e^{(\norm{M}+\rho)\,t-\mu\,d(x,y)}.
\]
\end{proof}

We specialize \Cref{lem:exponential-decay} to $X=E$ with the distance $d$ of \Cref{def:line-graph}, to $M =\Delta$ at time $2t$, and, assuming $\DD\ge 2$ so that $S_\mu>0$ (the case $\DD=1$ is settled directly in the proof of \Cref{thm:main}), to $\rho = 2S_\mu$. The bound \eqref{eq:improved-bound} becomes
\begin{equation}
    \bigl|e^{2t \Delta}(e,e')\bigr|
    \le \frac{2\norm{\Delta}+2\pi S_\mu}{\pi S_\mu}\;e^{2(\norm{\Delta}+2 S_\mu)\,t-\mu\,d(e,e')}.
\end{equation}

Each edge $e=\{x,y\}$ meets at most $(\deg x-1)+(\deg y-1)\le 2\mathcal D-2$ other edges, so $\deg_{L(G)}(e) \le 2\mathcal D-2 .$ Further $\|\Delta\|$ is at most its maximal row sum (Schur test, \cite[Proposition~10.6]{AizenmanWarzel}), so $\|\Delta\|\le 1+\sup_e \deg_{L(G)}(e)\le 2\mathcal D$.
Moreover, 
\[
  S_\mu=\sup_{e\in E}\sum_{e'\in E}\Delta(e,e')\bigl(e^{\mu d(e,e')}-1\bigr)
       =(e^\mu-1)\sup_{e\in E}\deg_{L(G)}(e)\le 2\mathcal D\,(e^\mu-1).
\] We use $\| \Delta \| \leq 2\DD$, $S_\mu \leq 2\DD(e^\mu-1)$ to obtain, from the previous bound:
\begin{equation}\label{eq:improved-delta}
    \bigl|e^{2t \Delta}(e,e')\bigr|
    \le \left( \dfrac{4\DD}{\pi S_\mu}+2 \right)\;e^{-4\DD t + 8\DD  e^\mu t-\mu\,d(e,e')}.
\end{equation}

\medskip
We are now ready to prove the main result. 

\begin{proof}[Proof of \Cref{thm:main}]
Since the commutator is bilinear, it suffices to prove the bound for $\|A\| = \|B\| = 1$, which we assume throughout. Since $G$ is connected with at least one edge, $\DD \geq 1$. If $\DD = 1$, then $G$ is a single edge, so $d(x,y) \le 1$ and the right-hand side of \eqref{eq:main-result} is at least $18/e > 2 \ge \|[A(t),B]\|$; hence there is nothing to prove, and we assume $\DD \geq 2$ henceforth.

Now, fix $t>0$ and edges $e,e'\in E$, and let $d(e,e')$ denote their distance in the line graph $L(G)$. 
Define $\Phi\colon\mathbb{R}\to\mathbb{R}$ as the part of the exponent in \eqref{eq:improved-delta} that depends on $\mu$:
\[
    \Phi(\mu) \;=\; 8\DD t\,e^{\mu} - \mu\,d(e,e').
\]
Then $\Phi$ is smooth with $\Phi''(\mu)=8\DD t\,e^{\mu}>0$, so $\Phi$ is
strictly convex; consequently any stationary point is its unique global
minimizer. Since
\[
    \Phi'(\mu) \;=\; 8\DD t\,e^{\mu} - d(e,e') ,
\]
the unique stationary point is
\begin{equation}\label{eq:phi-min}
    \mu_\ast \;:=\; \ln\!\left(\frac{d(e,e')}{8\DD t}\right),
    \qquad\text{and}\qquad
    \Phi(\mu_\ast) \;=\; d(e,e')\left(1 - \ln\!\left(\frac{d(e,e')}{8\DD t}\right)\right).
\end{equation}
Note that if $d(e,e') > 32\DD t$, then $e^{\mu_\ast} = d(e,e')/(8\DD t) > 4$, so
$\mu_\ast > \ln 4 > 0$; in particular $\mu_\ast$ lies in the admissible range
of the estimate.

\medskip
By \eqref{eq:improved-delta}, for every $\mu>0$,
\begin{equation}\label{eq:kernel}
    \bigl|e^{2t\Delta}(e,e')\bigr|
    \;\le\; \left(\frac{4\DD}{\pi S_\mu} + 2\right)
        \exp\!\bigl(-4\DD t + \Phi(\mu)\bigr),
    \qquad
    S_\mu \;
    \defeq\; \sup_{e\in E}\ \sum_{\substack{e'\sim e\\ e'\neq e}}\bigl(e^{\mu}-1\bigr),
\end{equation}
the exponent being exactly $-4\DD t+\Phi(\mu)$. Then, for $d(e,e')>32\DD t$, substituting $\mu=\mu_\ast$ gives:
\begin{align}
    \bigl|e^{2t\Delta}(e,e')\bigr| &
    \;\le\; \left(\frac{4\DD}{\pi S_{\mu_*}} + 2\right)
        \exp\!\bigl(-4\DD t + \Phi(\mu_*)\bigr) 
        \\
        & =\; \left(\frac{4\DD}{\pi S_{\mu_\ast}} + 2\right)
        e^{-4\DD t}\left(\frac{8e\DD t}{d(e,e')}\right)^{d(e,e')} \leq \left(\frac{4\DD}{\pi S_{\mu_\ast}} + 2\right)
        \left(\frac{8e\DD t}{d(e,e')}\right)^{d(e,e')}.
        \label{eq:optimized}
\end{align}
Since $S_{\mu_\ast}$ occurs in a denominator, it suffices to bound it from below.

\medskip
Every off-diagonal nonzero entry of the adjacency operator $\Delta$ equals
$1$, so the inner sum in \eqref{eq:kernel} has exactly
$\deg_{L(G)}(e)$ terms, each equal to $e^{\mu}-1$; therefore
\[
    S_\mu \;=\; \bigl(e^{\mu}-1\bigr)\,
        \sup_{e\in E} \deg_{L(G)}(e).
\]
The maximal degree of $L(G)$ is at least one less than the maximal degree of $G$. Moreover, assuming $d(e,e')>32\DD t$, the observation above that $e^{\mu_\ast} > 4$ gives
$e^{\mu_\ast}-1 > 4-1 = 3$. Combining the two,
\[
    S_{\mu_\ast} \;\ge\; 3(\DD-1),
    \qquad\text{which implies}\qquad
    \frac{4\DD}{\pi S_{\mu_\ast}} + 2 \;\le\; \frac{4}{3\pi}\left( \frac{\DD}{\DD - 1}\right) + 2 \;\le\; \frac{8}{3\pi} + 2\;<\; 3 ,
\]
since by assumption, $\DD \ge 2$. Substituting into \eqref{eq:optimized} and discarding the factor
$e^{-4\DD t}\le 1$, we obtain, for all $e,e'$ with $d(e,e')>32\DD t$,
\begin{equation}\label{eq:decay}
    \bigl|e^{2t\Delta}(e,e')\bigr|
    \;\le\; 3\left(\frac{8e\DD t}{d(e,e')}\right)^{d(e,e')}.
\end{equation}

\medskip
Let $A$ be supported at $x$ with $\|A\| = 1$, let $B$ be supported at $y$
with $\|B\| = 1$, fix an edge $e_1\ni y$, and let $e_0\ni x$ be an edge at which $(e^{2t\Delta})(e_1,\cdot)$ is largest among the edges containing $x$, as in \eqref{eq:qss-edge-recursion}. By
\eqref{eq:qss-edge-recursion},
\begin{equation}\label{eq:recursion}
    \bigl\|[A(t),B]\bigr\|
    \;\le\; 2\deg(x)\,\bigl(e^{2t\Delta}\bigr)(e_1,e_0).
\end{equation}
We now consider two regimes: suppose first that $d(e_1,e_0)>32\DD t$. Then \eqref{eq:decay} yields
\begin{equation}\label{eq:superexp-edge}
    \bigl\|[A(t),B]\bigr\|
    \;\le\; 6\deg(x)\left(\frac{8e\DD t}{d(e_1,e_0)}\right)^{d(e_1,e_0)} .
\end{equation}
Applying the elementary inequality $\bigl(1+\tfrac{s}{n}\bigr)^{n}\le e^{s}$,
valid for $n>0$ and $s>-n$, with $n=d(e_1,e_0)$ and
$s=8e\DD t-d(e_1,e_0)>-n$, gives
\[
    \bigl\|[A(t),B]\bigr\|
    \;\le\; 6\deg(x)\,e^{\,8e\DD t-d(e_1,e_0)} .
\]
Using $d(e_1,e_0)\ge d(x,y)-1$ (see \Cref{fig:line-graph}), $\deg(x)\le\DD$, $6e<18$ and
$8e<32$, we conclude
\begin{equation}\label{eq:main-bound}
    \bigl\|[A(t),B]\bigr\|
    \;\le\; 6\DD\,e^{\,8e\DD t-d(x,y)+1}
    \;\le\; 18\DD\,e^{\,8e\DD t-d(x,y)}
    \;\le\; 18\DD\,e^{\,32\DD t-d(x,y)} .    
\end{equation}
This proves \eqref{eq:main-result} in the first regime. As a by-product, we record the superexponential bound \eqref{eq:factorial} stated in \Cref{sec:lr-bound}: since $n \mapsto (8e\DD t/n)^{n}$ is decreasing for $n > 8\DD t$, and $d(e_1,e_0) \ge d(x,y) - 1$ with $\deg(x) \le \DD$, \eqref{eq:superexp-edge} gives, whenever $d(x,y) > 32\DD t + 1$ (which places us in the present regime),
\[
    \bigl\|[A(t),B]\bigr\|
    \;\le\; 6\DD\left(\frac{8e\DD t}{d(x,y)-1}\right)^{d(x,y)-1} ,
\]
which is \eqref{eq:factorial}.

\medskip
Suppose now $d(e_1,e_0)\le 32\DD t$. The trivial bound gives
$\|[A(t),B]\|\le 2\|A\|\,\|B\|\le 2$. On the other hand,
$d(x,y)\le d(e_1,e_0)+1\le 32\DD t+1$, so $32\DD t-d(x,y)\ge -1$, and since
$\DD\ge 1$,
\[
    18\DD\,e^{\,32\DD t-d(x,y)} \;\ge\; \frac{18}{e} \;>\; 2
    \;\ge\; \bigl\|[A(t),B]\bigr\| .
\]
Thus the bound \eqref{eq:main-result} holds in this regime as well, and therefore holds
for all $x,y$ and all $t>0$. This completes the proof. 
\end{proof}

\end{document}